\documentclass[english]{elsarticle}
\usepackage[T1]{fontenc}
\usepackage[latin9]{inputenc}
\usepackage{color}
\usepackage{float}
\usepackage{amsthm}
\usepackage{amsmath}
\usepackage{amssymb}
\usepackage{graphicx}
\usepackage{esint}

\makeatletter

\floatstyle{ruled}
\newfloat{algorithm}{tbp}{loa}
\providecommand{\algorithmname}{Algorithm}
\floatname{algorithm}{\protect\algorithmname}

\theoremstyle{plain}
\newtheorem{thm}{\protect\theoremname}
  \theoremstyle{definition}
  \newtheorem{defn}[thm]{\protect\definitionname}

\usepackage{url}

\usepackage{listings}
\usepackage{color}

\definecolor{mygreen}{rgb}{0,0.6,0}
\definecolor{mygray}{rgb}{0.5,0.5,0.5}
\definecolor{mymauve}{rgb}{0.58,0,0.82}

\renewcommand{\epsilon}{\varepsilon}

\makeatother

\usepackage{babel}
  \providecommand{\definitionname}{Definition}
\providecommand{\theoremname}{Theorem}

\begin{document}
\global\long\def\i{{\rm i}}
\global\long\def\e{\mathrm{e}}

\title{A splitting approach for the Kadomtsev--Petviashvili equation\tnoteref{label1}} \tnotetext[label1]{This work is  supported by the Austrian Science Fund (FWF) -- project id: P25346.}
\author[uibk]{Lukas Einkemmer\corref{cor1}} \ead{lukas.einkemmer@uibk.ac.at}
\author[uibk]{Alexander Ostermann}
\address[uibk]{Department of Mathematics, University of Innsbruck, Austria}
\cortext[cor1]{Corresponding author}
\begin{abstract}
	We consider a splitting approach for the  Kadomtsev--Petviashvili equation with periodic boundary conditions and show that the necessary interpolation procedure can be efficiently implemented. The error made by this numerical scheme is compared to exponential integrators which have been shown in Klein and Roidot (SIAM J. Sci. Comput., 2011) to perform best for stiff solutions of the Kadomtsev--Petviashvili equation. \textcolor{black}{Since many classic high order splitting methods do not perform well, we propose a stable extrapolation method in order to construct an efficient numerical scheme of order four. In addition, the conservation properties and the possibility of order reduction for certain initial values for the numerical schemes under consideration is investigated.}
\end{abstract} 
\maketitle

\section{Introduction \label{sec:intro}}

The Kadomtsev--Petviashvili equation (KP equation) is a model of nonlinear
wave propagation \textcolor{black}{which was proposed in \cite{kadomtsev1970}}; it is usually stated in the following form
\begin{equation}
\left(u_{t}+6uu_{x}+\epsilon^{2}u_{xxx}\right)_{x}+\lambda u_{yy}=0,\label{eq:KPequation}
\end{equation}
where $\lambda$ and $\epsilon$ are two parameters that are determined
by the physical problem under consideration. The KP equation appears
in the description of long wavelength waves, where we choose either
$\lambda=1$ (weak surface tension) or $\lambda=-1$ (strong surface
tension).
In accordance with the literature (see, for example, \citep{klein2011})
we call the latter the KP I model and the former the KP II model.

\textcolor{black}{The KP equation is a nonlinear dispersive partial differential equation that can be considered as a two-dimensional generalization of the well known Korte\-weg--de Vries equation (KdV equation). Similar to the KdV equation the KP equation is Hamiltonian and as a consequence does not include any dissipation. It exhibits many interesting physical phenomena  such as soliton solutions and blow-up in finite time (see, for example, \cite{klein2012}).}

Before a numerical scheme is applied equation (\ref{eq:KPequation})
is usually rewritten in evolution form
\begin{equation}
u_{t}+6uu_{x}+\epsilon^{2}u_{xxx}+\lambda\partial_{x}^{-1}u_{yy}=0,\label{eq:KP-evolution}
\end{equation}
where $\partial_{x}^{-1}$ is to be understood as the regularized
Fourier multiplier of $-\i/k_{x}$. That is, as in \citep{klein2011},
we impose periodic boundary conditions and use the Fourier multiplier
\[
\frac{-\i}{k_{x}+\i\lambda\delta},
\]
where $\delta$ is equal to machine epsilon (the smallest number that
in the finite precision arithmetic system under consideration yields
a result different from one when added to one). That is, for the double
precision floating point numbers employed in the simulations presented
here, we have $\delta=2^{-52}$.

The KP equation shows a number of interesting phenomena including
soliton solutions and the appearance of small scale oscillations.
For soliton solutions the stiffness of the KP equation is usually
only a minor concern (since the linear part of the equation can be solved by spectral methods). In this setting various types of IMEX methods
are usually very efficient. However, a number of phenomena do
display stiff behavior and therefore pose a significant challenge
for numerical schemes. In \citep{klein2011} it was found that in
this context exponential integrators, in many instances, outperform
IMEX and implicit Runge--Kutta methods. Furthermore, explicit time
integrators \textcolor{black}{suffer from a severe stability restriction of the time step size} (due to the third derivative that appears in the dispersive
term) that renders them computationally unfeasible.

In this paper we will demonstrate that splitting methods provide a
viable and computationally attractive alternative to exponential integrators
for stiff solutions of the KP equation. In section \ref{sec:Splitting-approach}
we introduce the Strang splitting approach and an exponential integrator
of order two. In section \ref{sec:Performance-considerations} the
performance of the Strang splitting scheme is compared to that of
the exponential integrator. The conservation properties of the splitting
approach are investigated in section \ref{sec:Conservation-properties}.
In section \ref{sec:High-order-splittings}
we consider both traditional high-order splitting schemes as well as a computationally attractive alternative approach based on Richardson extrapolation. Let us duly note that this approach avoids the stability problems often present if local Richardson extrapolation is applied to a nonlinear
problem. \textcolor{black}{We then provide, in section \ref{sec:pc}, a comparison of the run time between the second and fourth order methods for a given accuracy.} \textcolor{black}{In section \ref{sec:constraint} we consider order reduction that is observed for initial values which violate a constraint. In fact, we observe that the commonly observed order reduction is not present for the Strang splitting scheme.} Finally, we conclude in section \ref{sec:conclusion}.

\section{Numerical approach\label{sec:Splitting-approach}}

In this paper we will exclusively employ the setting described in
\citep{klein2011}. That is, the KP equation for a given initial value
and periodic boundary conditions is propagated in time. Within this
framework we are limited to initial values that are either periodic
or decrease sufficiently fast for large values of $\vert x\vert$
and $\vert y\vert$. 

In this setting, the form of equation (\ref{eq:KP-evolution}) suggests
an approach where the linear part can be solved very efficiently by
means of fast Fourier techniques. This eliminates the (severe) \textcolor{black}{stability
constraint} imposed by both the third and second order differential operators present
in the KP equation. Exponential integrators (see, for example, \citep{hochbruck2010})
exploit the fact that the linear part can be efficiently diagonalized.
Similar to Runge--Kutta methods, time integrators of arbitrary order
can be constructed where it is only required that the Burgers' nonlinearity
can be evaluated efficiently. However, while this method manages to
overcome a number of difficulties inherent in the numerical integration
of the KP equation, it also suffers from a number of disadvantages
due to the fact that the discretization of the Burgers' nonlinearity
is essentially explicit.
In fact, there is no mathematical proof that shows that exponential
integrators are stable for the KP equation. Note that such results
have been established for unbounded nonlinearities (using the parabolic
smoothing property) and for bounded nonlinearities (see, for example,
\citep{hochbruck2010}). From a numerical standpoint, such considerations
are important if the nonlinear dynamics is equally important as the
linear dynamics (for example, if $u$ is large in magnitude). In addition,
it is often not clear how exponential integrators behave with respect
to the conservation of invariants of the continuous system. For many
interesting problems splitting methods solve both of these problems.
In fact, stability (and convergence) of Strang splitting for a number
of dispersive equations with a Burgers' nonlinearity are available
in \citep{holden2013}.

In the splitting approach considered here, we compute an approximate
solution to the partial flows given by
\begin{equation}
u_{t}=Au=-\epsilon^{2}u_{xxx}-\lambda\partial_{x}^{-1}u_{yy}\label{eq:linear-flow}
\end{equation}
and
\begin{equation}
u_{t}=B(u)=-6uu_{x}.\label{eq:nonlinear-flow}
\end{equation}
If it is possible to efficiently compute sufficiently accurate approximations
to the partial flows given by equation (\ref{eq:linear-flow}) and
(\ref{eq:nonlinear-flow}), respectively, splitting methods constitute a viable approach. For example, the Strang splitting scheme
for the step size $\tau$ is given by
\[
u_{n+1}=\e^{\frac{\tau}{2}A}\left(\varphi_{\tau}^{B}\left(\e^{\frac{\tau}{2}A}u_{n}\right)\right),
\]
where (for a given initial value $u_{n}$) the linear partial flow
corresponding to equation (\ref{eq:linear-flow}) is denoted by $\mathrm{e}^{\tau A}u_{n}$
and the nonlinear partial flow corresponding to equation (\ref{eq:nonlinear-flow})
is denoted by $\varphi_{\tau}^{B}(u_{n})$. Let us also note that
since the linear half-steps can be combined, the Strang splitting
scheme does only need to compute the action of each partial flow once
during each time step.

\textcolor{black}{Before proceeding, let us note that the partial flow given in \eqref{eq:nonlinear-flow} is not well defined for arbitrarily large time steps. This is due to the fact that Burgers' equation develops a singularity for finite times. In principle this implies a step size restriction for the splitting approach. In all the simulations conducted in this paper the step size is determined by accuracy considerations only.  However, this would be a more serious concern in the small dispersion limit (i.e.~where $\epsilon\to0$). We will not consider this case here but remark that due to the large gradients in the solution, a different numerical procedure for both time and space discretization seems to be in order then.}

In the next section we will compare \textcolor{black}{the splitting} approach \textcolor{black}{outlined above} to the exponential
integrator of order two given by
\begin{align}
u_{n+1} & =\e^{\tau A}u_{n}+\tau\varphi_{1}(\tau A)B(u_{n})+\tau\varphi_{2}(\tau A)\left(B(U)-B(u_{n})\right),\label{eq:expint-order2}
\end{align}
where
\[
U=\e^{\tau A}u_{n}+\tau\varphi_{1}(\tau A)B(u_{n})
\]
and the $\varphi_{i}$ functions are given by the recurrence relation
\[
z\varphi_{k+1}(z)=\varphi_{k}(z)-\varphi_{k}(0)
\]
with initial value $\varphi_{0}(z)=\e^{z}$.

\section{Performance considerations\label{sec:Performance-considerations}}

In \citep{klein2011} it has been argued that the splitting approach
is not viable as the interpolation necessary to solve Burgers' equation
(i.e., to compute an approximation to the action of $\varphi_{\tau}^{B}$)
is too costly compared to the computation of $uu_{x}$ which only
requires two Fast Fourier Transforms (FFTs) as well as some complex
arithmetics. Therefore, we will consider this point in more detail
in this section.

The algorithm of Cooley and Tukey requires approximately $5n\log n$
floating point operations. Libraries, such as the Fastest Fourier
Transform in the West (FFTW \citep{FFTW}) used in our implementation,
provide very efficient implementations of the FFT (including optimizations
using SSE and AVX%
\footnote{The Streaming Single instruction, multiple data Extension (SSE) \textcolor{black}{and the Advanced Vector extensions (AVX) are}
a collection of CPU instructions that can be utilized to accelerate
code segments that applies the same operation to multiple floating
point numbers.%
} and the use of more advanced algorithms). On the other hand, using
the method of characteristics, we can derive an expression for the
exact solution of equation (\ref{eq:nonlinear-flow})
\begin{equation}
u_{1}(x)=u\left(0,x-6\tau u_{1}(x)\right),\label{eq:burgers-exact-representation}
\end{equation}
where $u_{1}(x)$ is the solution of (\ref{eq:nonlinear-flow}) at
time $\tau$ with initial value $u(0,\cdot)$. Note that for the KP
equation $y$ is a parameter in the above equation (that is, we have
to compute an approximation to $u_{1}(x)$ for each grid point in
the $y$-direction). The representation given here is still implicit
in $u_{1}$ and can be solved by conducting a fixed-point iteration.
In a practical numerical scheme, this fixed-point iteration has to
be truncated after a finite number of iterations (henceforth denoted
by $i$). Of course, the value of $i$ has a substantial impact on
the performance. The other ingredient necessary is an interpolation
algorithm. Such an algorithm is required as we have to determine the
value of $u(0,x_{i}-6\tau u_{1}(x_{i}))$, for each grid point $x_{i}$.
Let us further note that using the FFT algorithm is not possible in
this case as the translation does depend on $x_{i}$ itself; this
fact implies that the resulting points are no longer equidistant.
However, similar to semi-Lagrangian methods (see, for example, \citep{sonnendrucker1999})
we can use a (local) polynomial or a spline interpolation of sufficiently
high degree. 

Let us now consider the efficiency of constructing and evaluating
a spline approximation. Construction of a cubic spline requires $\mathcal{O}(n)$
(real) floating point operations (the cost of the tridiagonal matrix
solver). The resulting polynomial is accurate of order four. To evaluate
a polynomial then requires $4n$ floating point operations (where
we count one addition and one multiplication as one operation). Thus,
one would conclude that even for medium sized problems the floating
point operations count favors the spline interpolation. However, once
we consider an implementation in C++ the performance of this scheme
is somewhat disappointing. For example, using the GNU scientific library%
\footnote{on an Intel Core i5-3427 CPU and a problem of size $2^{11}\cdot2^{9}$.%
} (GSL \citep{GSL}) we need approximately $200$ ms to construct the
spline and $150$ ms for each fixed-point iteration. 

On the other
hand performing two FFTs (as is required to compute the Burgers' nonlinearity)
requires only $120$ ms. The second order exponential integrator,
in total, requires the evaluation of two nonlinearities and an additional
$6$ FFTs for the computation of the matrix functions, yielding a
total cost of approximately $600$ ms per time step, whereas the Strang
splitting algorithm requires approximately $320+i\cdot150$ ms per
time step. Thus the Strang splitting scheme has approximately equal
cost if we choose $i=2$ (a value that is presumably too small). We
have also used the ALGLIB library and found its performance significantly
worse than GSL. \textcolor{black}{Note that the GSL library also offers an
interpolation that constructs and evaluates the interpolation polynomial for
a number of specified grid points. We have used this to perform interpolation
with a stencil of four grid points that is centered at the evaluation point.
However, using this}
polynomial approximation approach
from GSL does not significantly improve performance either; even though
in this case we do not have to construct a global spline.

The issue here is not only one of optimization (GSL is most certainly
not as well tuned as FFTW is) but in fact does relate to the problem
that is being solved. The FFT algorithm must assume that it operates
on an equidistant grid. This is not true for a spline or polynomial
interpolation. In fact, all libraries require both an array of the
grid points and the function values. Also, the FFT expansion is global
thus alleviating the need for (a possible expensive) modulo operation
in order to determine which part of the approximation needs to be
accessed. Furthermore, GSL does not know a priori that we only employ
fixed degree polynomials. Thus, it has to implement an algorithm that
is stable even if high degree polynomials need to be constructed on
a highly irregular grid. None of the difficulties listed above (except
for performing a modulo operation) are relevant here. Based on the
Lagrange form, we have implemented a cubic approximation (i.e., an
approximation of order four) that only requires $60\cdot i$ ms. 
\textcolor{black}{Note that for this implementation the cost of the cubic
interpolation is equal to the cost of performing a single FFT.}
This then means that $i=8$ would yield a Strang splitting scheme that
is equal in execution time to the exponential integrator of order
two. The details of this implementation are given in \textcolor{black}{\ref{sec:Efficient-cubic}}.

Now, at least two questions remain to be answered: what value of $i$
is required in order to obtain a sufficiently accurate approximation
and how does the error of the Strang splitting scheme compare to the
second order exponential integrator. To that end, we have conducted
numerical experiments for the KP~I and KP II equations using the
Schwartzian initial value (as is done in \citep{klein2011}, for example)
given by
\begin{equation}
u(0,x,y)=-\tfrac{1}{2}\partial_{x}\text{sech}^{2}\left(\sqrt{x^{2}+y^{2}}\right).\label{eq:schwartzian-iv}
\end{equation}
 The results are shown in Figure \ref{fig:KPI-schwartzian} (KP I
equation) and Figure \ref{fig:KPII-schwartzian} (KP II equation).
We observe that the Strang splitting scheme yields an error that is
smaller by a factor of $10$ for the KP I equation and smaller by
a factor of $3$ for the KP II equation. These gains can be realized
by only performing three fixed-point iterations. The increase in accuracy
together with the very competitive run-time leads to the conclusion
that splitting methods can in fact be very competitive in the setting
considered.

\begin{figure}
\begin{centering}
\includegraphics[width=10cm]{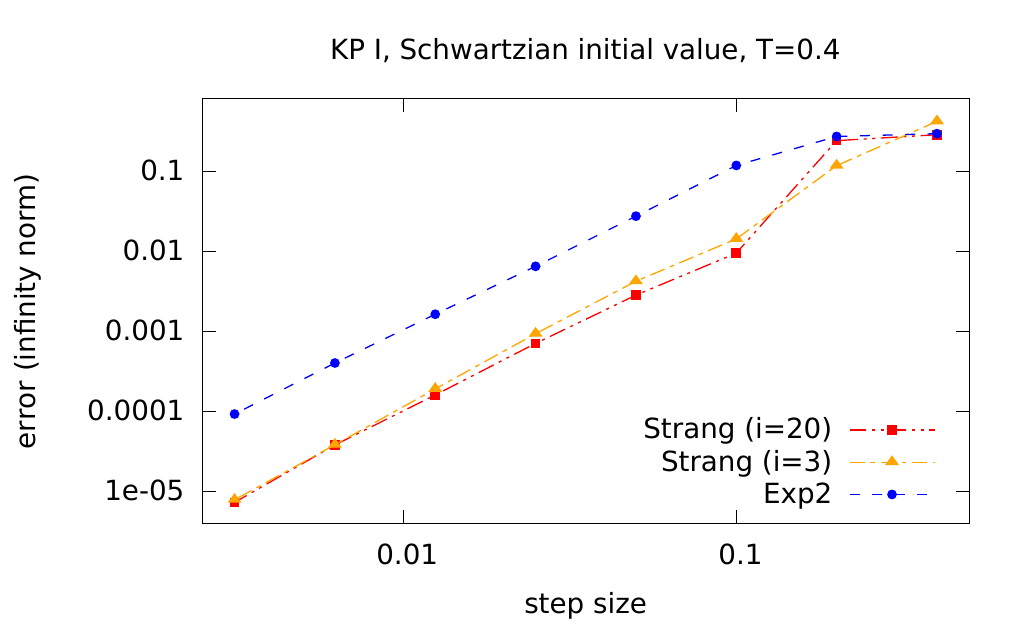}
\par\end{centering}

\protect\caption{The error (in the infinity norm) as a function of the step size is
shown at time $t=0.4$ for the KP I equation using the Schwartzian
initial value (\ref{eq:schwartzian-iv}). The parameter $\epsilon$
is chosen equal to $0.1$. To discretize space we have employed $2^{11}$
grid points in the $x$-direction and $2^{9}$ grid points in the
$y$-direction (on a domain of size $[-5\pi,5\pi]\times[-5\pi,5\pi]$).
The number of iterations conducted to solve Burgers' equation for
the Strang splitting scheme is denoted by $i$ and the exponential
integrator (\ref{eq:expint-order2}) of order two is referred to as
Exp2. \textcolor{black}{The error is computed using a reference solution with step size equal to $10^{-3}$.}\label{fig:KPI-schwartzian}}
\end{figure}

\begin{figure}
\begin{centering}
\includegraphics[width=10cm]{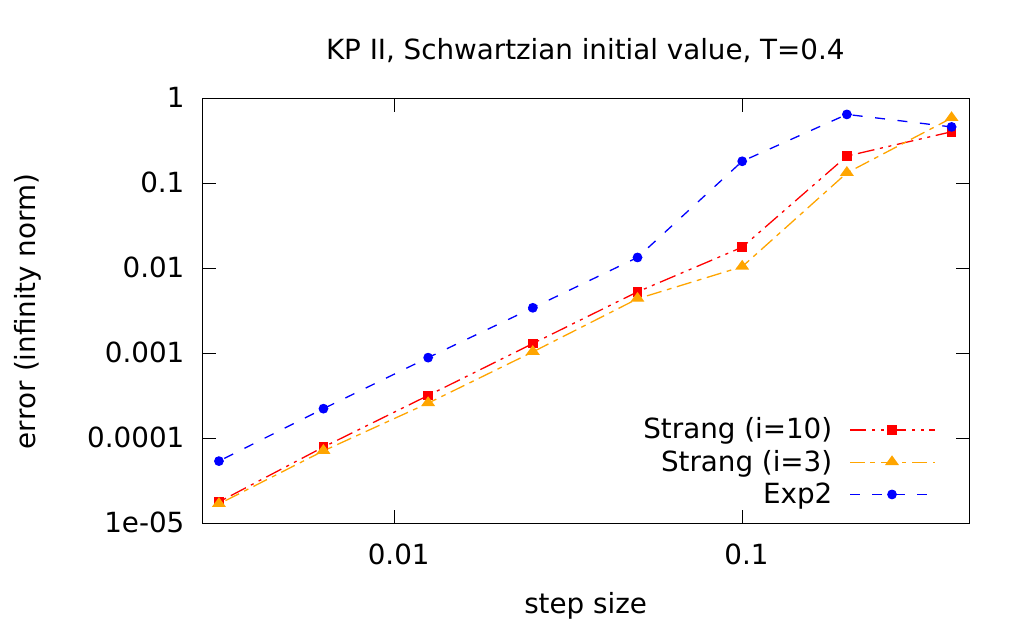}
\par\end{centering}

\protect\caption{The error (in the infinity norm) as a function of the step size is
shown at time $t=0.4$ for the KP II equation using the Schwartzian
initial value (\ref{eq:schwartzian-iv}). The parameter $\epsilon$
is chosen equal to $0.1$. To discretize space we have employed $2^{11}$
grid points in the $x$-direction and $2^{9}$ grid points in the
$y$-direction (on a domain of size $[-5\pi,5\pi]\times[-5\pi,5\pi]$).
The number of iterations conducted to solve Burgers' equation for
the Strang splitting scheme is denoted by $i$ and the exponential
integrator (\ref{eq:expint-order2}) of order two is referred to as
Exp2. \textcolor{black}{The error is computed using a reference solution with step size equal to $10^{-3}$.} \label{fig:KPII-schwartzian}}
\end{figure}

The results of the relative performance between the Strang splitting
method and the exponential integrator of order two can be understood
by considering the relative strength of the dispersive term $\epsilon u_{xxx}$
and the Burgers' nonlinearity $6uu_{x}$. For the Schwartzian initial
value the Burgers' nonlinearity is larger in magnitude by approximately
a factor of $10$. As time evolves dispersive effects eventually take
over. This happens more slowly in the case of the KP I equation than
for the KP II equation, which in turn explains the larger gain in
accuracy achieved by the splitting approach in the former case (it
is expected that splitting methods provide increased relative performance,
compared to exponential integrators, as the importance of the Burgers'
nonlinearity increases).

Let us further note that, as stated in \citep{klein2011}, the analysis
conducted above is strictly speaking only correct if the Fourier multipliers
can be precomputed. This holds true for a constant step size integrator
but not if adaptive step size control is employed. In fact, recomputing
the Fourier multipliers (due to the complex exponential) is by at
least a factor of $5$ more costly than performing the forward and
backward FFT. In the Strang splitting scheme this only affects a single
exponential while in the exponential integrator of order two the Fourier
multiplier for two additional $\varphi$ functions have to be recomputed.

\textcolor{black}{To conclude this section let us note that in performing the
splitting algorithm spectral convergence is lost. This is due to the fact
that we employ a polynomial interpolation in solving the Burgers' equation
which is only of order four. For a fixed time step size the error in space
is shown in Figure \ref{fig:KPI-space}.}

\begin{figure}
	\textcolor{black}{
	\centering
	\includegraphics[width=10cm]{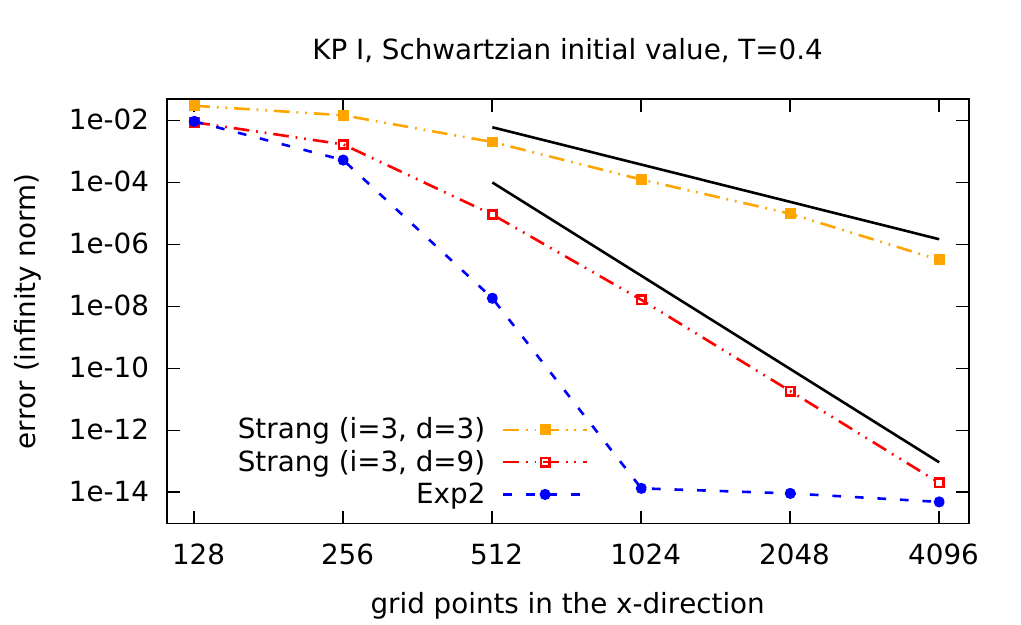}
	\caption{The error (in the infinity norm) as a function of the number of grid points in the $x$-direction is shown at time $t=0.4$ for the KP I equation using the Schwartzian
initial value (\ref{eq:schwartzian-iv}). The parameter $\epsilon$
is chosen equal to $0.1$. We employ $2^{9}$ grid point in the $y$-direction and a fixed
time step size that is equal to $\tau=10^{-2}$. The domain under consideration is of size $[-5\pi,5\pi]\times[-5\pi,5\pi]$. The number of iterations conducted to solve Burgers' equation for
the Strang splitting scheme is denoted by $i$ and the degree of the polynomial interpolation is denoted by $d$. The exponential integrator (\ref{eq:expint-order2}) of order two is referred to as
Exp2. The error is computed using a reference solution with $2^{13}$ grid points in the $x$-direction and a line of slope $4$ and $10$ is shown for comparison. \label{fig:KPI-space}}}
\end{figure}

\textcolor{black}{
	Let us note that while spectral convergence in space is certainly a desirable property, it has to be considered in the context of the time discretization error. Certainly there is no point in using a space discretization that is exact up to machine precision while making a time discretization error on the order of $10^{-2}$. If a high accuracy in space is required polynomials of higher degree can be used. For example, the ninth degree polynomial interpolation shown in Figure \ref{fig:KPI-space} is approximately three times as costly as the cubic interpolation. Nevertheless, for $i=3$ the Strang splitting scheme is still almost twice as fast compared to the exponential integrator of order two (for an equal number of grid points).
}

\section{Conservation properties\label{sec:Conservation-properties}}

In addition to using a scheme of sufficient accuracy at minimal computational
cost, it is often desirable to employ a method that conserves certain
invariants of the continuous problem (in this case the KP equation).
This both ensures a physically consistent solution and usually facilitates
the long time integration. It has long been known that an infinite
number of quantities is conserved by the KP equation \citep{lin1982}.
\textcolor{black}{Note, however, that most of the high order invariants
are only formal. That is, they are not well defined on suitable function
spaces (see, e.g.~\cite{molinet2007}).}

\textcolor{black}{In this paper only linear and quadratic invariants are considered that have a clear physical interpretation}. Following \citep{minzoni1996} these are the linear
invariant $m(t)$ (corresponding to mass)
\[
m(t)=\int_{\Omega}u(t,x,y)\,\mathrm{d}(x,y)
\]
and the quadratic invariant $M(t)$ (corresponding to momentum)
\[
M(t)=\int_{\Omega}u(t,x,y)^{2}\,\mathrm{d}(x,y).
\]
In addition, for the KP equation the constraint 
\begin{equation} \label{eq:constraint}
	\textcolor{black}{ \int_{-\infty}^{\infty}\partial_{yy}u(t,x,y)\,\mathrm{d}x=0 }
\end{equation}
is satisfied. This property, however, is respected for both the Strang
splitting scheme as well as the exponential integrator up to machine
precision \textcolor{black}{(a consequence of the regularization)}.

Since Runge--Kutta methods preserve linear invariants (such as the
mass in the KP equation) we might expect that the same holds true
for exponential Runge--Kutta methods (all of the exponential integrators
considered in this paper are in fact exponential Runge--Kutta methods).
A more formal definition (see \citep{hochbruck2010}) is given in
Definition \ref{def:exponential-Runge--Kutta}.
\begin{defn}
\label{def:exponential-Runge--Kutta}A exponential Runge--Kutta method
applied to $u^{\prime}=Au+B(t,u)$, $u(0)=u_{0}$ is given by
\begin{eqnarray*}
u_{1} & = & \e^{\tau A}u_{0}+\tau\sum_{i=1}^{s}b_{i}(\tau A)G_{i},\\
U_{i} & = & \e^{c_{i}\tau A}u_{0}+\tau\sum_{j=1}^{s}a_{ij}(\tau A)G_{j},\\
G_{j} & = & B(c_{j}\tau,U_{j}),
\end{eqnarray*}
where $u_{1}$ is an approximation to $u(\tau)$. The method is said
to have $s\in\mathbb{N}$ stages and is uniquely determined by the
coefficients $c_{i}\in\mathbb{R}$ and the coefficient functions $b_{i}$
and $a_{ij}$, where $i,j\in\left\{ 1,\dots,s\right\} $. The functions
$b_{i}$ are assumed to be linear combinations of $\varphi_{k}$ functions.
\end{defn}
 In contrast to Runge--Kutta methods, we have to assume that the
linear invariant under consideration is conserved for both the flow
generated by $A$ and the flow generated by $B$. This assumption
is satisfied for the KP equation. In the following theorem we assume
that $A$ and $B$ already have been discretized in space in such
a way that the linear invariant considered is a conserved quantity
of the discretized system. It should, however, be duly noted that
the proof of Theorem \ref{thm:conservation-expRK} can just as well
be carried out for the case where space is left continuous.
\begin{thm}
\label{thm:conservation-expRK}An exponential Runge--Kutta method
preserves every linear invariant that is preserved by both the flow
generated by $A$ and the flow generated by $B$.\end{thm}
\begin{proof}
Since the $b_{i}$ are linear combination of $\varphi_{k}$ functions
(which, in general, will be evaluated for different step sizes), in
order to show that $d$ is an invariant of the numerical method we
have to show that
\[
d^{\mathrm{T}}u_{1}=d^{\mathrm{T}}u_{0}.
\]
For the exponential Runge--Kutta method we have
\[
d^{\mathrm{T}}u_{1}=d^{\mathrm{T}}u_{0}+\tau\sum_{i=1}^{s}d^{\mathrm{T}}b_{i}(hA)G_{i},
\]
since $d^{\mathrm{T}}\e^{\tau A}u_{0}=d^{\mathrm{T}}u_{0}$.

Now, we will show that $v(\tau)=\tau^{k}\varphi_{k}(\tau A)g$ satisfies
\begin{equation}
v^{\prime}(\tau)=Av(\tau)+\frac{\tau^{k-1}}{k!}g,\qquad v(0)=0.\label{eq:phifunction-diffequation}
\end{equation}
 Let us recall the recurrence relation for the $\varphi_{k}$ function
\[
\varphi_{k+1}(\tau A)g=(\tau A)^{-1}\left(\varphi_{k}(\tau A)-\varphi_{k}(0)\right)g
\]
for which upon multiplication by $\tau^{k+1}$ and differentiating
with respect to time we get
\begin{align*}
\partial_{\tau}\tau^{k+1}\varphi_{k+1}(\tau A)g & =A^{-1}\partial_{\tau}\left(\tau^{k}\varphi_{k}(\tau A)-\tau^{k}\varphi_{k}(0)\right)g\\
 & =\tau^{k}\varphi_{k}(\tau A)g+\tau^{k-1}A^{-1}\left(\frac{1}{(k-1)!}-k\varphi_{k}(0)\right)g\\
 & =A(\tau A)^{-1}\left(\tau^{k+1}\varphi_{k}(\tau A)-\tau^{k+1}\varphi_{k}(0)\right)g+\frac{\tau^{k}}{k!}g\\
 & =A\left(\tau^{k+1}\varphi_{k+1}(\tau A)g\right)+\frac{\tau^{k}}{k!}g.
\end{align*}
A simple calculation in the case for $\varphi_{1}$ completes the
induction.

Since we can assume that both $d^{\mathrm{T}}Aw=0$ and $d^{\mathrm{T}}g=0$
for any $w$ and $g$, we immediately follow from equation (\ref{eq:phifunction-diffequation})
that
\[
d^{\mathrm{T}}v(t)=d^{\mathrm{T}}v(0)=0
\]
which implies that
\[
d^{\mathrm{T}}\tau b_{i}(\tau A)G_{i}=0.
\]
This completes the proof.
\end{proof}
Before continuing let us note that all the methods considered in this
paper or in \citep{klein2011} satisfy the assumption on the coefficient
functions $b_{i}$ given in Definition~\ref{def:exponential-Runge--Kutta}.
Furthermore, since we employ a FFT based discretization in space,
which conserves the mass exactly, we expect that the exponential integrator
considered here do in fact conserve the mass (up to machine precision).

Now, let us numerically investigate the conservation of mass. To that
end we perform simulations of the KP I and the KP II equation using
the Schwartzian initial value up to the final time $t=2$. A slice
of the solution (for $y=0$) is shown in Figure \ref{fig:KPI2} (for
the KP I equation) and in Figure \ref{fig:KPII2} (for the KP II equation).
\begin{figure}
\begin{centering}
\includegraphics[width=10cm]{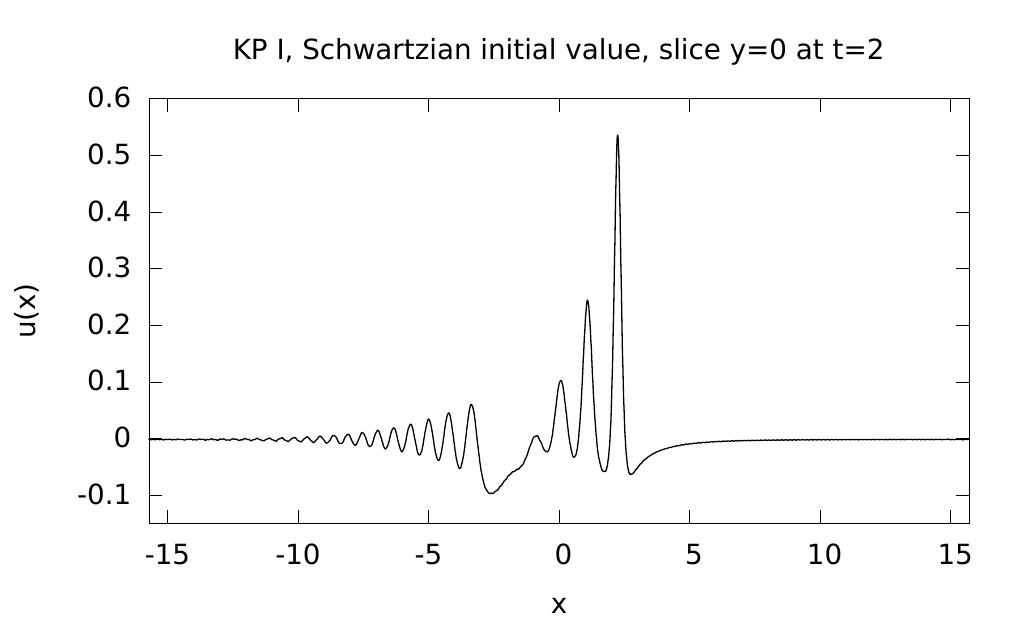}
\par\end{centering}

\protect\caption{A one-dimensional slice (at $y=0$) of the numerical solution of the
KP I equation for the Schwartzian initial value (\ref{eq:schwartzian-iv})
at $t=2$ is shown. To discretize space $2^{11}$ grid points are
employed in the $x$-direction and $2^{9}$ grid points are employed
in the $y$-direction (on a domain of size $[-5\pi,5\pi]\times[-5\pi,5\pi]$).
\label{fig:KPI2}}
\end{figure}
\begin{figure}
\begin{centering}
\includegraphics[width=10cm]{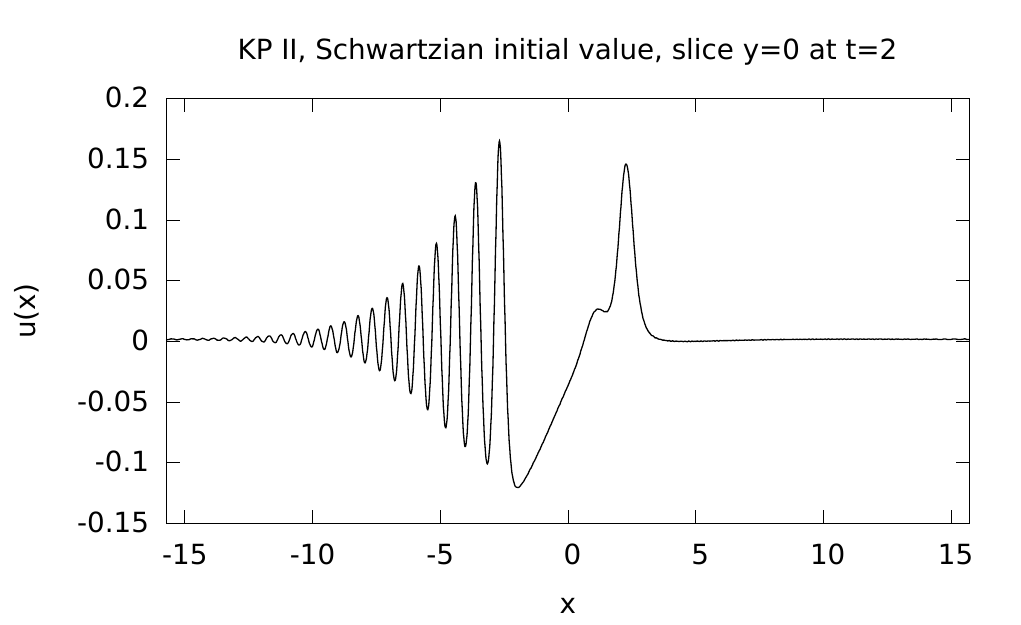}
\par\end{centering}

\protect\caption{A one-dimensional slice (at $y=0$) of the numerical solution of the
KP II equation for the Schwartzian initial value (\ref{eq:schwartzian-iv})
at $t=2$ is shown. To discretize space $2^{11}$ grid points are
employed in the $x$-direction and $2^{9}$ grid points are employed
in the $y$-direction (on a domain of size $[-5\pi,5\pi]\times[-5\pi,5\pi]$).\label{fig:KPII2}}
\end{figure}

The error in the mass, that is $\vert m(t)-m(0)\vert$, is shown as
a function of time in Figure \ref{fig:massKPI} (for the KP I equation)
and in Figure \ref{fig:massKPII} (for the KP II equation). We observe,
as expected from the theoretical result, that the exponential integrator
conserves the mass up to machine precision, while for the Strang splitting
scheme \textcolor{black}{the error ranges from $10^{-6}$ to $10^{-10}$ depending on the number of iterations $i$ performed and the number of grid points used.} From the perspective
of the splitting approach this behavior seems to be disappointing
and perhaps contrary to intuition. However, it is entirely expected
since by using the \textcolor{black}{cubic polynomial} interpolation we no longer have exponential
convergence (as is the case for the Fourier approximation) and consequently
a projection error is made in computing a solution to Burgers' equation.
As we can see from Figure \ref{fig:massKPI} the error does depend
(weakly) on the number of fixed-point iterations conducted. Note that
if we employ a finer space discretization (and increase the number
of iterations) then the error in mass of the splitting approach does
decrease as well (see Figures \ref{fig:massKPI} and \ref{fig:massKPII}).

It is, however, not clear what the ramifications for long time integration
are. In the context of semi-Lagrangian methods this was studied in
some detail. It was found that even though the mass in such interpolation
methods is not exactly conserved, they remain remarkably stable over
long times (see, for example, \citep{filbet2003}). 

\begin{figure}
\begin{centering}
\includegraphics[width=10cm]{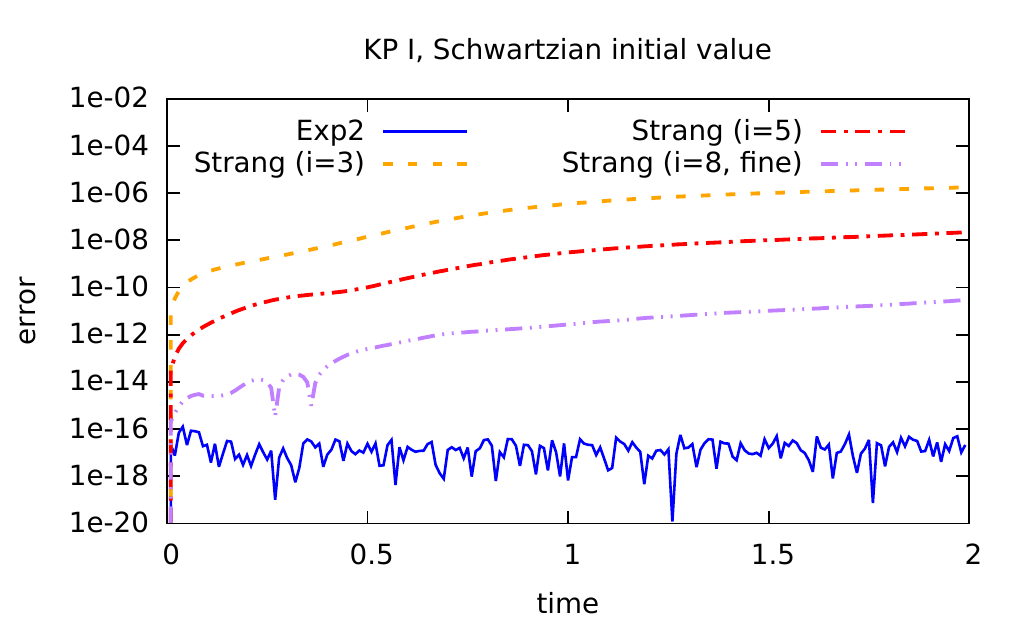}
\par\end{centering}

\protect\caption{The error in mass, i.e. $\vert m(t)-m(0)\vert$ is shown as a function
of time for the KP I equation. A time step of size $10^{-2}$ is used.
To discretize space $2^{11}$ grid points are employed in the $x$-direction
and $2^{9}$ grid points are employed in the $y$-direction (on a
domain of size $[-5\pi,5\pi]\times[-5\pi,5\pi]$), except for the
fine discretization in which case $2^{13}\times2^{9}$ grid points
are used. \label{fig:massKPI}}
\end{figure}
\begin{figure}
\begin{centering}
\includegraphics[width=10cm]{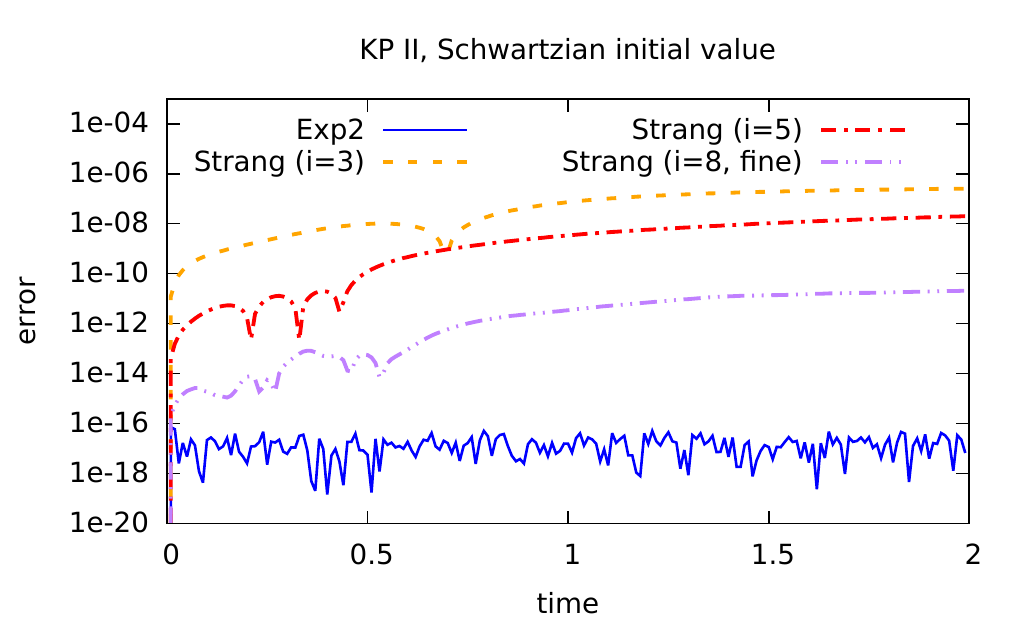}
\par\end{centering}

\protect\caption{The error in mass, i.e. $\vert m(t)-m(0)\vert$ is shown as a function
of time for the KP II equation. A time step of size $10^{-2}$ is
used. To discretize space $2^{11}$ grid points are employed in the
$x$-direction and $2^{9}$ grid points are employed in the $y$-direction
(on a domain of size $[-5\pi,5\pi]\times[-5\pi,5\pi]$), except for
the fine discretization in which case $2^{13}\times2^{9}$ grid points
are used. \label{fig:massKPII}}
\end{figure}

Now, let us consider the conservation of momentum. The results of
the numerical solutions are shown in Figure \ref{fig:energyKPI} (for
the KP I equation) and Figure \ref{fig:energyKPII} (for the KP II
equation). In the former case we observe that for $2^{11}\times2^{9}$
grid points the Strang splitting scheme is more accurate by two orders
of magnitude, while in the latter case only a difference of one order
of magnitude in accuracy is observed. Contrary to the exponential
integrator, where the time step size has a significant impact on the
conserved quantities, we observe a decrease in the error in \textcolor{black}{mass}
as the number of grid points is increased. To obtain these results,
we also have to slightly increase the number of iterations. The additional
iterations performed, as compared to the order plots presented in
the last section, do not appreciably \textcolor{black}{decrease} the error in \textcolor{black}{mass} but do result in better conservation properties (if a sufficiently
fine space discretization is used).

\begin{figure}
\begin{centering}
\includegraphics[width=10cm]{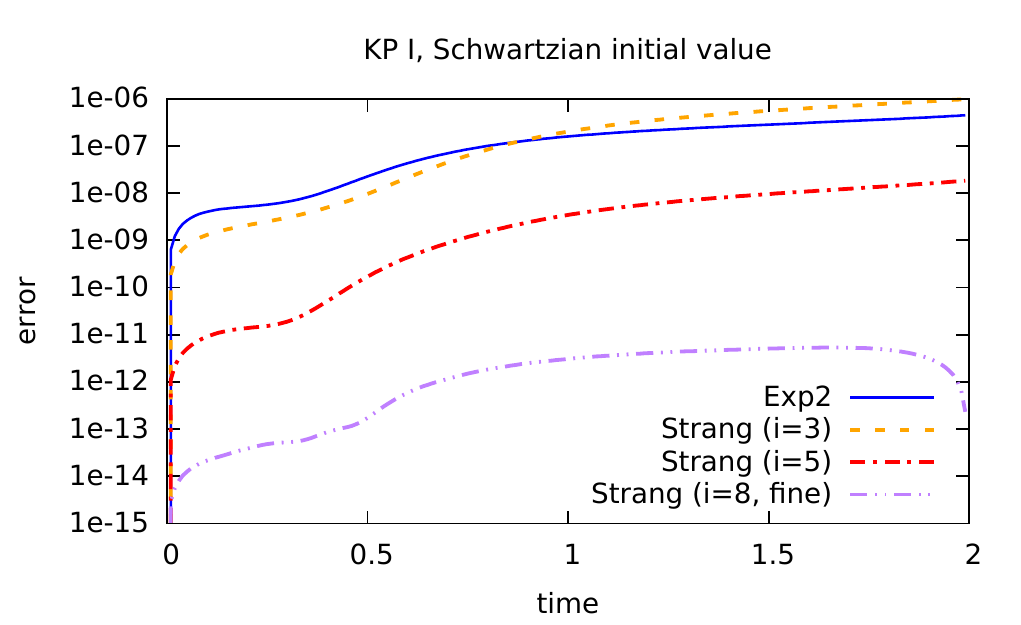}
\par\end{centering}

\protect\caption{The error in momentum, i.e. $\vert M(t)-M(0)\vert$ is shown as a
function of time for the KP I equation. A time step of size $10^{-2}$
is used. To discretize space $2^{11}$ grid points are employed in
the $x$-direction and $2^{9}$ grid points are employed in the $y$-direction
(on a domain of size $[-5\pi,5\pi]\times[-5\pi,5\pi]$), except for
the fine discretization in which case $2^{13}\times2^{9}$ grid points
are used. \label{fig:energyKPI}}
\end{figure}
\begin{figure}
\begin{centering}
\includegraphics[width=10cm]{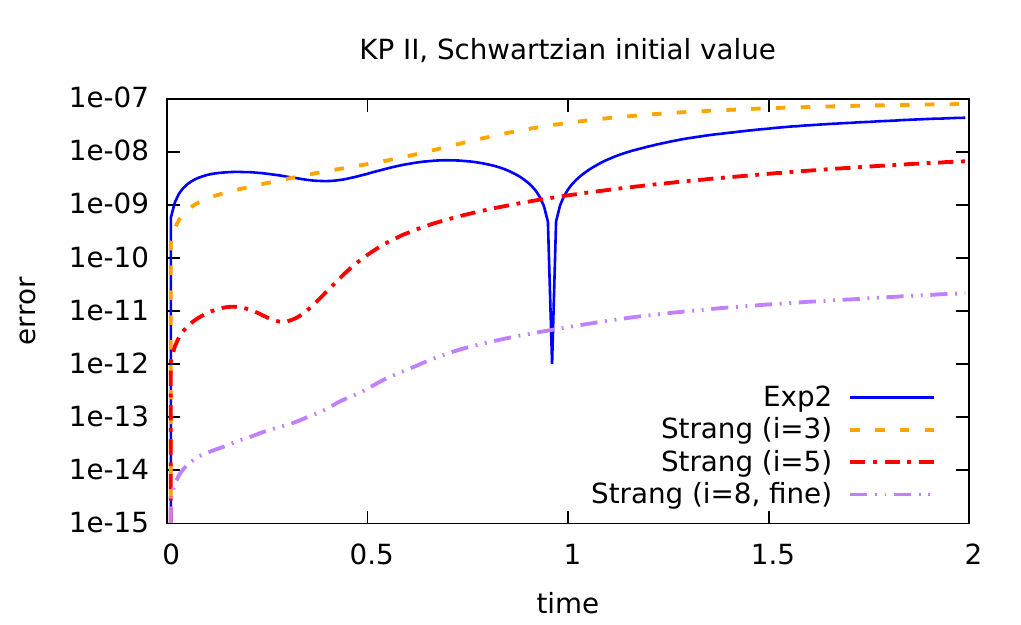}
\par\end{centering}

\protect\caption{The error in the momentum, i.e. $\vert M(t)-M(0)\vert$ is shown as
a function of time for the KP~II equation. A time step of size $10^{-2}$
is used. To discretize space $2^{11}$ grid points are employed in
the $x$-direction and $2^{9}$ grid points are employed in the $y$-direction
(on a domain of size $[-5\pi,5\pi]\times[-5\pi,5\pi]$), except for
the fine discretization in which case $2^{13}\times2^{9}$ grid points
are used. \label{fig:energyKPII}}
\end{figure}

\section{High order splitting\label{sec:High-order-splittings}}

An mathematical rigorous result (see, for example, \citep{blanes2005}) shows
that if real time steps are assumed and if the order of a splitting
method is strictly larger than two, both partial flows have to be
computed for a step size that is smaller than zero (i.e., we have
to conduct steps backward in time). An alternative, see \citep{hansen2009},
is to use complex time steps (with positive real part). The former
can be used in purely hyperbolic partial differential equations to
obtain methods of arbitrary order, while the latter can be used in
purely parabolic partial differentials equations (with some performance
penalty due to the necessity of using complex arithmetics) to obtain
methods of high order. 

{\color{black}
The KP equation is hyperbolic and its eigenvalues are purely imaginary. Therefore, in principle, employing a splitting approach with negative time steps is a possibility. However, due to the regularization introduced in section \ref{sec:intro} taking negative time steps results in an exponential amplification of round-off errors.

To explain this behavior let us consider the initial value problem
\[ u_t(t,x,y) = \partial_x^{-1}u_{yy}(t,x,y)\]
which, after regularization, yields the following equations in Fourier space 
\[ \hat{u}_t(t,k_x,k_y) = \frac{-i}{k_x+i\lambda \delta} k_y^2 \hat{u} = \frac{-\lambda \delta  - i k_x}{k_x^2+(\lambda \delta)^2} k_y^2 \hat{u}(t,k_x,k_y). \]
Since these equations are decoupled in the wavevectors $(k_x,k_y)$ we can consider $k_x=0$ for an arbitrary wavevector $k_y\neq0$. The solution is then given by
\begin{equation} \label{eq:ampl}
	\hat{u}_t(t,0,k_y) = \e^{-t k_y^2/(\lambda \delta)}\hat{u}(0,0,k_y)
\end{equation}
which for $t<0$ implies that any non-zero value in any of these components is exponentially amplified which is clearly a undesirable behavior for any implementation in finite precision. 

If the initial value does not satisfy the constraint given in equation \eqref{eq:constraint} this regularization exponentially damps the modes with $k_x=0$. This is in fact the expected behavior of the underlying continuous model. However, even if the constraint is satisfied by the initial value numerical errors (such as those made in the approximation of Burgers' equation) result in a non-zero value for these modes. Such a behavior, however, is unphysical and will eventually pollute the numerical solution. Thus, we need some mechanism to dissipate these  modes if high order splitting schemes are to be applied.

Since the factor in the exponential of equation \eqref{eq:ampl} immediately sets the modes under consideration to zero (even for very small time steps), we propose to implement this behavior implicitly in the numerical scheme. The corresponding solution is, up to machine precision, equivalent to the regularization procedure introduced in the introduction (which is used extensively in the literature). Note that the mode with $k_x=0$ and $k_y=0$ is constant in time and thus no regularization is required. At least for initial values that satisfy constraint \eqref{eq:constraint} this allows us to perform negative time steps (the value of the modes under consideration is neither changed by a positive nor by a negative time step).
}

{\color{black}
While the Strang splitting scheme is the universally employed second order splitting scheme, a variety of of different fourth order splitting schemes have been proposed. The often used triple jump scheme is relatively cheap from a computational point of view. Its implementation is only three times as expensive as an implementation of the Strang splitting scheme. Unfortunately, the triple jump schemes employs negative time steps of length $1.7 \tau$. This, implies that we require additional iterations in order to obtain good accuracy for the solution of Burgers' equation. However, even if this is done, numerical simulation suggest that the error constants of this method is very disappointing.

The length of the negative time step necessary for fourth order splitting methods can be reduced by considering additional stages. For example, the methods of Suzuki and McLachlan (see, for example, \cite{hairer2006}) require roughly five times the computational effort compared to the Strang splitting scheme. In our numerical simulations the method of McLachlan has been found to perform best for both the KP I and KP II equation. However, on drawback of the method by McLachlan is that it is somewhat expensive from a computational point of view.
}


\textcolor{black}{As an alternative, we consider} the so-called Richardson extrapolation algorithm
which enables the construction of higher order methods from an (almost)
arbitrary numerical one-step method $S_{\tau}$ with step size $\tau$.
In the following we limit ourselves to the case where $S_{\tau}$
is a method of order two. We proceed by performing a step with length
$\tau$ and two steps with $\tau/2$. Then, the final approximation
$u_{n+1}$ is computed from $u_{n}$ as follows
\[
u_{n+1}=\frac{4S_{\tau/2}\left(S_{\tau/2}(u_{n})\right)-S_{\tau}(u_{n})}{3}.
\]
This procedure eliminates the leading error term in $S_{\tau}$ and
due to the symmetry of the Strang splitting scheme results in a method
which is consistent of order four. \textcolor{black}{Note that since symmetry can be defined for the semi-discrete case (i.e.~after space has already been discretized) order four is achieved independent of the space discretization under consideration.}
However, in general, the resulting scheme is not stable in the nonlinear case. In fact, we observe this
lack of stability for the KP equation. 

Therefore, we propose to apply a global extrapolation algorithm (see
\citep{verwer1985}). First, we compute
\[
v_{n+1}=S_{\tau}(v_{n})
\]
and 
\[
w_{n+1}=S_{\tau/2}\left(S_{\tau/2}(w_{n})\right),
\]
where $v_{0}$ and $w_{0}$ are equal to the initial value $u_{0}$.
Then, we compute the final approximation, for each time step $n$,
as follows:
\[
u_{n}=\frac{4w_{n}-v_{n}}{3}.
\]
Note that this is in fact the Richardson extrapolation algorithm.
But instead of applying it at each time step, we first compute a solution
with time step $\tau$ and a solution with time step $\tau/2$ and
then apply the extrapolation procedure independently for each time
step. This alleviates the stability problems as both $v_{n}$ and
$w_{n}$ are computed by the same unconditionally stable scheme (but
using a different step size). 

In \citep{klein2011} a number of exponential integrators have been
compared in the context of the KP equation. It was found that the
method of Cox and Matthews \citep{cox2002}, the method of Krogstad
\citep{krogstad2005}, and the method of Hochbruck and Ostermann \citep{hochbruck2010}
do show almost identical performance characteristics (even though
they differ in run-time as well as accuracy). We have chosen to compare
the extrapolation scheme described in this section \textcolor{black}{and the method by McLachlan} with the method
of Cox and Matthews. The results are shown in Figure \ref{fig:KPI-extrapolation}
(for the KP I equation) and Figure \ref{fig:KPII-extrapolation} (for
the KP II equation). \textcolor{black}{We observe that in case of the KP I equation
the method of McLachlan is more accurate by approximately an order of magnitude compared both to the extrapolation method and the method of Cox and Matthews. For the KP II equation the
method of Cox and Matthews performs better and is is as accurate as the method of McLachlan. Both methods are superior to the extrapolation method by approximately a factor of $7$.} Note that the last statement is only true in the asymptotic case. There is a region (up to an error of approximately $10^{-2}$)
where the extrapolation method is more accurate.

The method of Cox and Matthews requires $4$ evaluations of the nonlinearities
and the computation of $12$ matrix functions. This gives a total
of $32$ FFTs that have to be performed which compared to the second
order method increases the cost by a factor of $3.2$. This is almost
the same increase in cost by a factor of $3$ which is required for
the extrapolation method. \textcolor{black}{However, the method of McLachlan is $2.5$ times as expensive as the extrapolation method. For a fourth order method this would require a gain in accuracy of almost a factor of $40$ for the method to be competitive. Let us note, however, that the method of McLachlan shows better conservation properties compared to the extrapolation approach.} 

\begin{figure}
\centering{}\includegraphics[width=10cm]{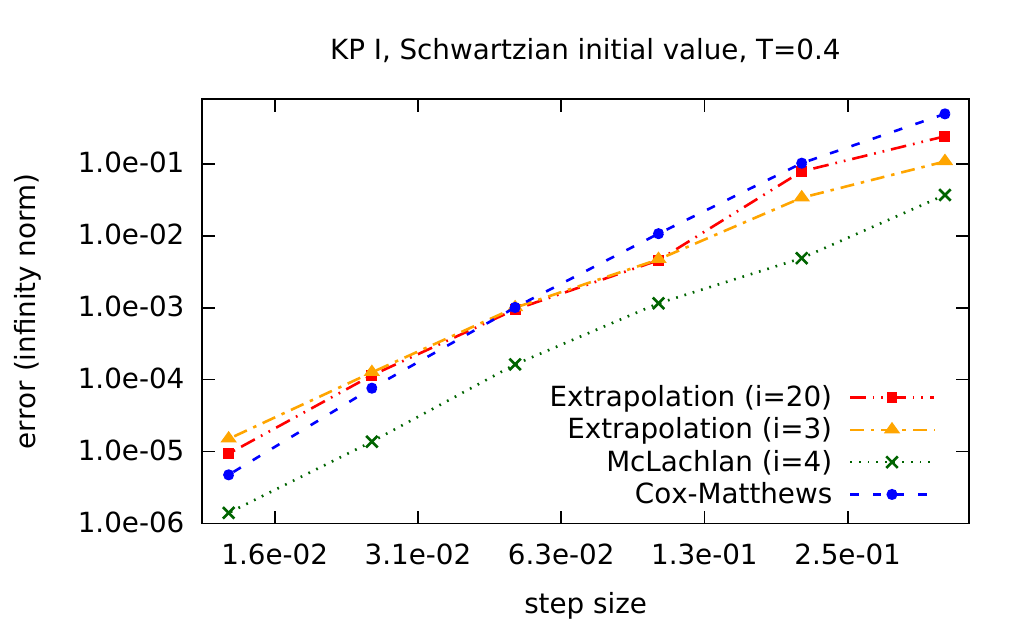}\protect\caption{The error (in the infinity norm) as a function of the step size $\tau$
is shown at time $t=0.4$ for the KP I equation using the Schwartzian
initial value (\ref{eq:schwartzian-iv}). The parameter $\epsilon$
is chosen equal to $0.1$. To discretize space we have employed $2^{11}$
grid points in the $x$-direction and $2^{9}$ grid points in the
$y$-direction (on a domain of size $[-5\pi,5\pi]\times[-5\pi,5\pi]$).
The number of iterations conducted to solve Burgers' equation for
the extrapolation scheme and McLachlan's method is denoted by $i$. We have used the exponential
integrator developed by Cox and Matthews. \textcolor{black}{The error is computed using a reference solution with step size equal to $10^{-3}$.}\label{fig:KPI-extrapolation}}
\end{figure}

\begin{figure}
\centering{}\includegraphics[width=10cm]{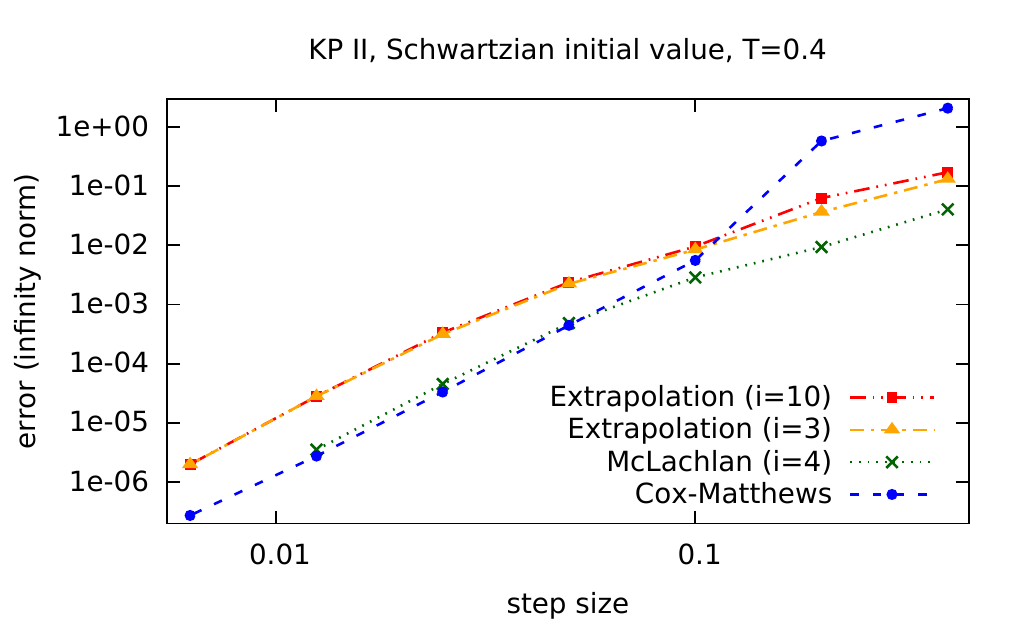}\protect\caption{The error (in the infinity norm) as a function of the step size is
shown at time $t=0.4$ for the KP II equation using the Schwartzian
initial value (\ref{eq:schwartzian-iv}). The parameter $\epsilon$
is chosen equal to $0.1$. To discretize space we have employed $2^{11}$
grid points in the $x$-direction and $2^{9}$ grid points in the
$y$-direction. The number of iterations conducted to solve Burgers'
equation for the extrapolation scheme and McLachlan's method is denoted by $i$. We have
used the exponential integrator developed by Cox and Matthews (on
a domain of size $[-5\pi,5\pi]\times[-5\pi,5\pi]$). \textcolor{black}{The error is computed using a reference solution with step size equal to $10^{-3}$.} \label{fig:KPII-extrapolation}}
\end{figure}

\textcolor{black}{Furthermore, we have analyzed the conservation properties of the fourth order methods considered in this section in case of the KP I equation. As before we integrate the equation until final time $T=2$. In this setting the method of McLachlan is the most robust scheme preserving momentum up to $10^{-10}$ for a time step size between $10^{-2}$ and $3\cdot10^{-2}$ (using the fine space discretization and $i=8$). In this case lack of conservation is only due to the error made in solving Burgers' equation. On the other hand the extrapolation scheme shows a behavior similar to the method of Cox and Matthews. While for small time steps conservation to high accuracy can be observed  (using a time step size of $\tau = 10^{-2}$ we observe an error in momentum approximately equal to $2\cdot 10^{-12}$ for the extrapolation scheme and $6\cdot 10^{-9}$ for the method of of Cox and Matthews), even if the time step size is only increased to $3 \cdot 10^{-2}$ the error in momentum increases to $2\cdot 10^{-7}$ for the extrapolation scheme and to $3\cdot 10^{-7}$ for the method of Cox and Matthews. The behavior of the extrapolation scheme is due to the fact that the error in momentum is now limited by the non-conservative nature of the extrapolation procedure.}

{\color{black}
\section{Performance comparison \label{sec:pc}}

In the previous sections we have only considered the error as a function of the time step size. Together with the performance considerations given in section \ref{sec:Performance-considerations}, we are able to compare the relative performance of the second order exponential integrator and the Strang splitting scheme. However, it is difficult to compare a second to a fourth order scheme. In addition, it is instructive to compare the various schemes in terms of the run time that is necessary to achieve a given accuracy in time. 

The purpose of this section is to perform the corresponding comparison. We employ the same initial values that are considered in section \ref{sec:Performance-considerations}. The numerical results consider a tolerance between $10^{-1}$ and $10^{-5}$ and are shown in Figure \ref{fig:rt-KPI} (for the KP I equation) and Figure \ref{fig:rt-KPII} (for the KP II equation). For convenience the speedup of using the Strang splitting approach (which is superior to the extrapolation scheme for the accuracy considered here) to the best exponential integrator (either the second order exponential integrator or the method of Cox and Matthews) is indicated for a tolerance of $10^{-2}$ and $10^{-3}$.

We observe that for low accuracy the Strang splitting scheme is faster by a factor of $3$ to $7$ compared to the exponential integrators. This is true for both the KP I and KP II equations. In addition, we observe that even for relatively low accuracy the method of Cox and Matthews is superior to the exponential integrator of order two considered here. On the other hand, the extrapolation method only overtakes the splitting method for a tolerance of approximately $10^{-5}$. Let us further note that the performance of the extrapolation method is always better than that of the method by Cox and Matthews (although the difference between the two methods for accuracies below $10^{-5}$ is negligible). Note that even if we employ a ninth degree polynomial interpolation (instead of the cubic interpolation considered so far) the performance of the extrapolation method and the exponential integrator of Cox and Matthews is almost equal (the extrapolation method is faster for the KP I equation and the method of Cox and Matthews is slightly faster for the KP II equation).

%
%
\begin{figure}
	\includegraphics[width=10cm]{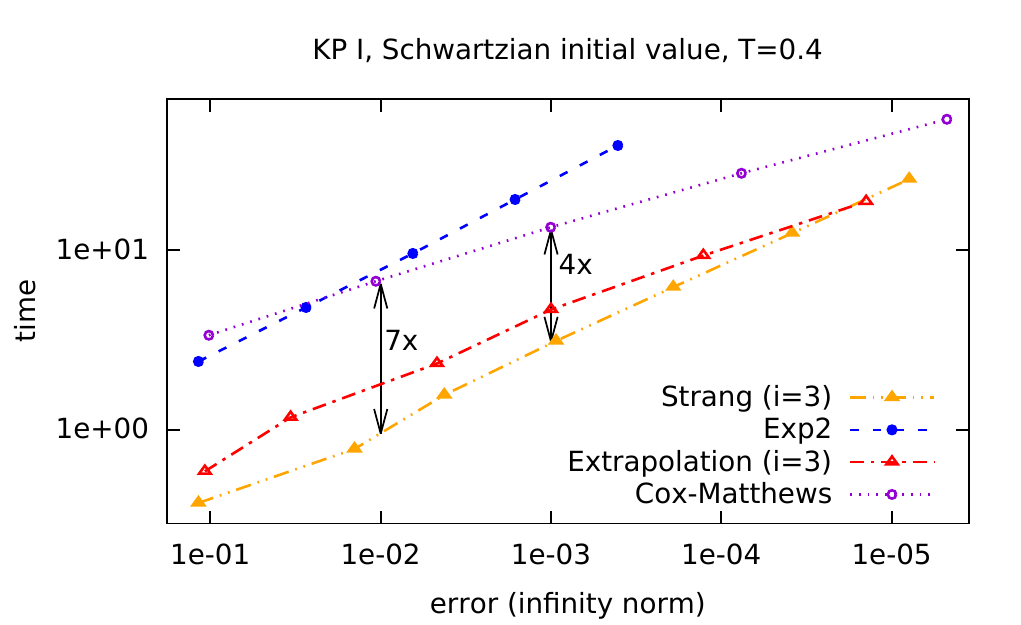}
	
	\caption{\color{black}The error (in the infinity norm) as a function of the run time is
shown at time $t=0.4$ for the KP I equation using the Schwartzian
initial value (\ref{eq:schwartzian-iv}). The parameter $\epsilon$
is chosen equal to $0.1$. To discretize space we have employed $2^{11}$
grid points in the $x$-direction and $2^{9}$ grid points in the
$y$-direction (on a domain of size $[-5\pi,5\pi]\times[-5\pi,5\pi]$).
The number of iterations conducted to solve Burgers' equation for
the Strang splitting scheme is denoted by $i$ and the exponential
integrator (\ref{eq:expint-order2}) of order two is referred to as
Exp2. The error is computed using a reference solution with step size equal to $10^{-3}$. \label{fig:rt-KPI}}
\end{figure}

\begin{figure}
	\includegraphics[width=10cm]{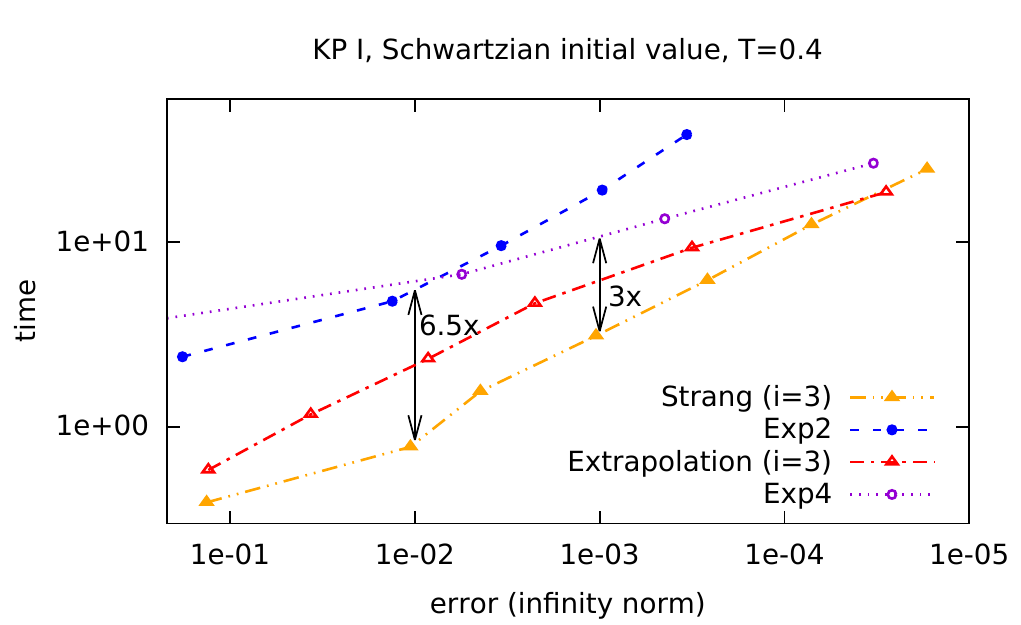}

	\caption{\color{black}The error (in the infinity norm) as a function of the run time is
shown at time $t=0.4$ for the KP II equation using the Schwartzian
initial value (\ref{eq:schwartzian-iv}). The parameter $\epsilon$
is chosen equal to $0.1$. To discretize space we have employed $2^{11}$
grid points in the $x$-direction and $2^{9}$ grid points in the
$y$-direction (on a domain of size $[-5\pi,5\pi]\times[-5\pi,5\pi]$).
The number of iterations conducted to solve Burgers' equation for
the Strang splitting scheme is denoted by $i$ and the exponential
integrator (\ref{eq:expint-order2}) of order two is referred to as
Exp2. The error is computed using a reference solution with step size equal to $10^{-3}$. \label{fig:rt-KPII}}
\end{figure}

}

{\color{black}
\section{Initial values that violate a constraint \label{sec:constraint}}

It is well known that the KP equation does satisfy the constraint \eqref{eq:constraint} for positive times $t>0$ even if this is not the case for the initial value (see, for example, \cite{fokas1999} and \cite{molinet2007}). This behavior is enforced by a discontinuity in time for the continuous problem and by the regularization for the discrete problem. For time integration schemes this usually results in order reduction (see, for example, \cite{klein2007}).

To investigate this phenomenon we will consider the initial value
\begin{equation} \label{eq:iv-no}
	u(0,x,y) = \alpha \mathrm{e}^{-D(x^2+y^2)}
\end{equation}
which does not satisfy the constraint and 
\begin{equation} \label{eq:iv-yes}
	u(0,x,y) = \beta x \mathrm{e}^{-D(x^2+y^2)}
\end{equation}
which satisfies the constraint. For both initial values $D=0.5$ is used. The parameters $\alpha=0.35$ and $\beta=0.6$ have been chosen such that the maximal amplitude of both initial values is comparable.

The result of our numerical experiments are shown in Figure \ref{fig:constr}. For the second order exponential integrator we observe order reduction in case of the initial value that does \textit{not} satisfy the constraint (even though this seems to be the easier problem). On the other hand, for the Strang splitting scheme the numerical results are consistent with a numerical method of order two (i.e.~no severe order reduction is present). These results (as shown in Figure \ref{fig:constr}) imply that for the evaluation of the Gaussian pulse the Strang splitting scheme is more than two orders of magnitude more accurate compared to the exponential integrator of order two (for the same step size).

\begin{figure}
	\centering
	\includegraphics[width=10cm]{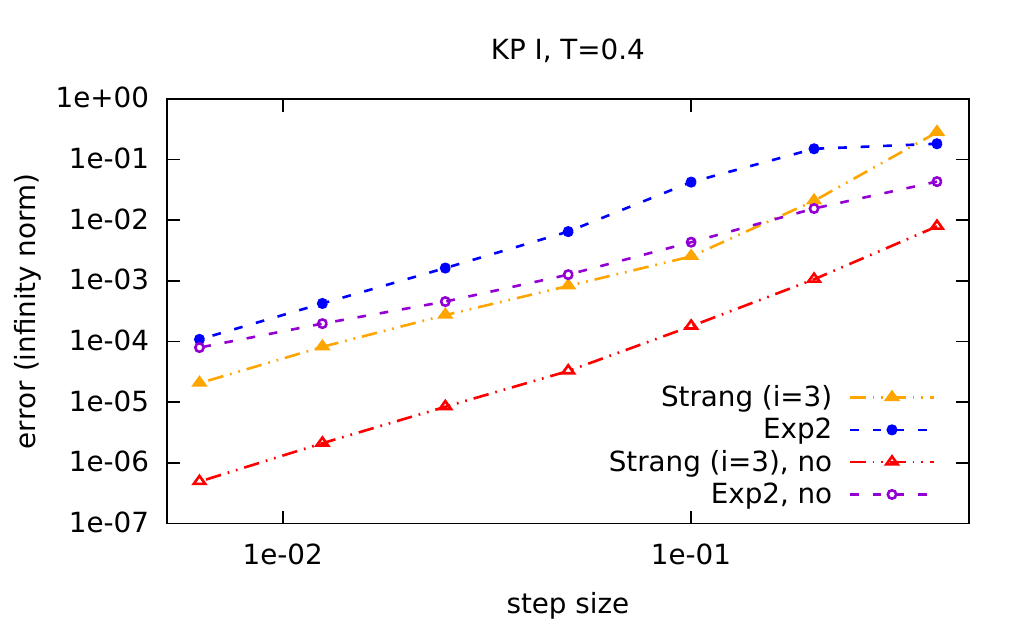}

	\caption{\color{black}The error (in the infinity norm) as a function of the step size is
shown at time $t=0.4$ for the KP I equation using the initial values
given in equation (\ref{eq:iv-no}) and (\ref{eq:iv-yes}). The numerical results for the former initial value, which does not satisfy the constraint is denoted by an additional \textit{no} in the legend of the plot. The parameter $\epsilon$ is chosen equal to $0.1$. To discretize
space we have employed $2^{11}$ grid points in the $x$-direction
and $2^{9}$ grid points in the $y$-direction (on a domain of size
$[-5\pi,5\pi]\times[-5\pi,5\pi]$). The number of iterations conducted
to solve Burgers' equation for the Strang splitting scheme is denoted
by $i$ and the exponential integrator of order two is referred to
as Exp2. The error is computed using a reference solution with step size equal to $10^{-3}$.\label{fig:constr}}

\end{figure}

}

\section{Conclusion \& Outlook\label{sec:conclusion}}

We have demonstrated that splitting methods for the KP equation can be efficiently implemented and achieve performance that \textcolor{black}{is significantly better compared to state of the art time integrators}. The efficient implementation of the projections necessary for computing an approximation to the solution of Burgers' equation, as demonstrated in \ref{sec:Efficient-cubic}, yields a Strang splitting scheme that, in addition to improved accuracy, is less expensive (as in most situations we only have to compute a few fixed-point iterations) compared to the exponential integrator
of order two. \textcolor{black}{In addition, conservation of momentum is exact for the time integrator. A challenge for such methods is the error made by the space approximation which no longer shows spectral convergence in the number of grid points. Also conservation of mass and momentum is influenced by the polynomial approximation in space. We consider this as further research.}

\textcolor{black}{We have also considered high order splitting schemes (for a modified regularization) and proposed an extrapolation method for Strang splitting. The latter method is computationally attractive and achieves comparable accuracy to the exponential integrator of Cox and Matthews for the KP I equation but worse accuracy for the KP II equation. However, the good conservation of momentum observed for the Strang splitting scheme is lost by conducting the extrapolation. Note that the fourth order method by McLachlan shows the best conservation properties among the methods considered in this paper. It is, however, more expensive computationally.}

\textcolor{black}{Let us note that,} for applications which require long time integration neither the exponential integrators (due to their conservation properties) nor the splitting approach (due to the lack of efficient high order methods) provide an ideal numerical scheme. \textcolor{black}{We consider this as further research.}

\bibliographystyle{plainnat}
\bibliography{KPequation}

\appendix

\section{Efficient implementation of cubic polynomial interpolation\label{sec:Efficient-cubic}}

The implementation shown here is based on the Lagrange form of the
interpolation polynomial through four equidistant nodes. One important
aspect of the algorithm is to determine the integer and fractional
part of the evaluation point with respect to the numerical grid. Usually
this would involve modulo operations which, however, significantly
impact performance. Therefore we only use casting to integer and replace
modulo operations by arithmetic operations, wherever possible. The
algorithm is divided into two loops; this approach yields a performance
gain of $1.5$ compared to the monolithic implementation (using the
GCC C++ compiler). The code is shown in Algorithm \ref{alg:lagrange}.

\begin{algorithm}
\begin{lstlisting}[language=C++]
// Evaluates a cubic Lagrange polynomial on the grid {-1,0,1,2}.
double lagrange3(double x,double u0m1, double u0p0, double u0p1,
					double u0p2) {
	double lm1 = -0.16666666666666667*x*(x-1.0)*(x-2.0);
	double l0  =  0.5*(x+1.0)*(x-1.0)*(x-2.0);
	double l1  = -0.5*(x+1.0)*x*(x-2.0);
	double l2  =  0.16666666666666667*(x+1.0)*x*(x-1.0);
	return u0m1*lm1 + u0p0*l0 + u0p1*l1 + u0p2*l2;
}

// u0: the initial value, u1: the result of the computation
// nx/ny: number of grid points in x/y-direction, L: domain length
// fp_it: number of fixed-point iterations conducted
void burgers(array2d& u0, array2d& u1, double h, double L, int nx,
				int ny, int fp_it=3) {

	// Some values that can be precomputed.
	double adv = 6.0*h*double(nx)/L;
	vector<double> d_i(nx);
	for(int i=0; i<nx; i++) 
		d_i[i] = double(i);
	vector<double> xred(nx);

	// Iteration over the y-direction.
	for(int j=0; j < ny; j++) {

		// The restrict keyword tells the compiler that there is no
		// pointer aliasing. This is essential for vectorization.
		double* __restrict _xred  = &xred[0];

		// The fixed-point iteration
		for (int k = 0; k < fp_it; k++) {

			// Avoid copying u0 to u1 in the first iteration.
			double* __restrict _us = (k==0) ? &u0(0,j) : &u1(0,j);

			// Determine an array of positions in [0,L].
			for(int i = 0; i < nx; i++) {
				double x = d_i[i] - adv*_us[i];
				if(x < 0)   x += nx;  // We assume that |u|<L/(6*h) holds.
				if(x >= nx) x -= nx;  // We assume that |u|<L/(6*h) holds.
				_xred[i] = x;
			}

			const double * __restrict _u0 = &u0(0,j);
			      double * __restrict _u1 = &u1(0,j);
			for (int i = 0; i < nx; i++) {
				// Determine the interpolation nodes.
				int p0 = (int)_xred[i];
				if(p0 == nx)  p0 = 0;
				int pm1 = p0-1;
				if(pm1 == -1)  pm1 = nx-1;
				int p1 = p0+1;
				if(p1 == nx)  p1 = 0;
				int p2 = p1+1;
				if(p2 == nx)  p2 = 0;

				// Evaluate the Lagrange interpolation.
				double x = _xred[i] - double(p0);
				_u1[i] = lagrange3(x, _u0[pm1], _u0[p0], _u0[p1], _u0[p2]);
			}
		}
	}
}

\end{lstlisting}

\protect\caption{Computation of the fixed-point iteration used to solve (\ref{eq:burgers-exact-representation})
(including the construction and evaluation of the Lagrange interpolation
polynomial). \label{alg:lagrange}}
\end{algorithm}

\end{document}